\documentclass[11pt]{article}
\pdfoutput=1
\usepackage[left=1in,right=1in,top=1in,bottom=1in]{geometry}
\usepackage{latexsym}
\usepackage{amsmath}
\usepackage{amssymb}
\usepackage{graphicx}
\usepackage{url}
\usepackage[belowskip=-5pt,aboveskip=0pt]{caption}

\newtheorem{theorem}{Theorem}
\newtheorem{lemma}[theorem]{Lemma}

\newenvironment{proof}{\noindent\textbf{Proof: }\ignorespaces}
  {\hspace*{\fill}$\Box$\medskip}

\def\A{{\mathcal{A}}}

\title{De-amortizing Binary Search Trees}
\author{Prosenjit Bose, S\'ebastien Collette, Rolf Fagerberg and
  Stefan Langerman}

\author{
Prosenjit Bose\thanks{School of Computer Science, Carleton
University. Email: \texttt{jit@scs.carleton.ca}. Research supported in
part by NSERC.}
\and
S\'ebastien Collette\thanks{Charg\'e de recherches du F.R.S.-FNRS,
  Département d'Informatique, Universit\'e Libre de Bruxelles. 
Email: \texttt{secollet@ulb.ac.be}.}
\and
Rolf Fagerberg\thanks{Department of Mathematics and Computer Science,
University of Southern Denmark. Email: \texttt{rolf@imada.sdu.dk}.
Partially supported by
the Danish Council for Independent Research, Natural Sciences.}
\and 
Stefan Langerman\thanks{Ma{\^\i}tre de recherches du F.R.S.-FNRS,
  Département d'Informatique, Universit\'e Libre de Bruxelles. 
Email: \texttt{stefan.langerman@ulb.ac.be}.}
}

\date{}

\begin{document}
\maketitle
\abstract{
We present a general method for de-amortizing essentially any
Binary Search Tree (BST) algorithm. In particular, by transforming Splay
Trees, our method produces a BST that has the same asymptotic cost as Splay
Trees on any access sequence while performing each search in $O(\log n)$
worst case time.  By transforming Multi-Splay Trees, we obtain a BST that
is $O(\log \log n)$ competitive, satisfies the scanning theorem, the static
optimality theorem, the static finger theorem, the working set theorem,
and performs each search in $O(\log n)$ worst case time.
Moreover, we prove that if there is a dynamically optimal BST algorithm,
then there is a dynamically optimal BST algorithm that answers every search
in $O(\log n)$ worst case time.
}

\thispagestyle{empty}

\newpage

\setcounter{page}{1}

\section{Introduction}

Over half a century since the discovery of rotation-based Binary Search
Trees, their exact performance is still not fully understood.
The very first works on BST focused on maintaining the tree balanced
($O(\log n)$ height and search time) after performing insertions and
deletions  \cite{avl,redblack}, or guaranteeing better average case bounds
for searches with known distributions \cite{optimum1}.

By introducing splay trees \cite{splay}, Sleator and Tarjan proposed an
alternate view of the problem, where instead of looking at the cost of
individual searches, it is the entire cost of a sequence of accesses
which is bounded, using amortized analysis.

The purpose of this article is to show that the two approaches are not
exclusive---i.e., that it is possible to combine the good amortized
performances of self-adjusting and other adaptive BST with strong worst
case guarantees for individual searches.

\paragraph{The BST Model.}
In order to describe accurately our results, we choose one BST model among
several existing standard variants, most of which are asymptotically
equivalent. In line with previous work, we will not consider insertions and
deletions. Hence, our BST model consists of a binary search tree $T$
containing the $n$ distinct keys $\{1,2,\ldots,n\}$ with their natural
order. The position of a finger, initially at the root of $T$, is
maintained, and the following two \emph{BST operations}, each of unit cost,
are allowed: 1) moving the finger from a node to its parent or to one of
its children, and 2) performing a rotation between the node pointed to by
the finger and its parent.

Given the current tree~$T$ and the current finger position, an
\emph{access} to a key $x$ is a list of BST operations (finger movements
and rotations), during which the finger position is at the node containing
$x$ at least once.

For an input sequence $S=\langle s_1,s_2,\ldots,s_m\rangle$ of keys to be accessed, a
BST algorithm $\A$ that \emph{realizes} $S$ returns a list $\A(S)$ of BST
operations for accessing the keys $s_1, s_2,\ldots$ in that order---that
is, where $S$ is a subsequence of the sequence of keys pointed to by the
finger during the execution of $\A(S)$.  An \emph{offline} algorithm $\A$
is given the entire sequence $S$ and the starting tree $T$ as input and
then outputs the sequence of operations $\A(S)$, while an \emph{online}
algorithm is fed the keys from~$S$ one by one and must output the BST
operations for the access of one key before the next key is given. More
formally, A is online if $\A(S)$ is a prefix of $\A(S')$ whenever $S$ is a
prefix of $S'$. The \emph{cost} of $\A(S)$ is the number of BST operations
it contains.

Note that the model, as all the standard variants of the BST model used in
competitive analysis of online BST algorithms, only requires the algorithm
to list the BST operations~$\A(S)$ to be performed (see, e.g, \cite{wilber89}). In particular, the
model does not restrict how those operations are generated, what auxiliary
memory is used in order to generated them, or even how much time is used to
generate them.

Of course, real-world implementations of practical BST algorithms have some
sensible limits on their time and space usage. In fact, almost all BST
implementations in the literature besides adhering to the standard BST
model described above
also have the following additional features: they work in the pointer
machine model, use no more space than the tree itself plus $O(1)$ words of
balance information in each node of the tree and $O(1)$ extra working
variables, and generate their access sequence~$\A(S)$ in time proportional
to the BST model cost of $\A(S)$. The majority of this paper is devoted to
showing how to de-amortize BST algorithms, with a method working in the
standard BST model. As a final step, we show how to extend the method to
maintain the additional features just listed, should the BST algorithm
being de-amortized have these.
 
Denote by $OPT$ the best offline algorithm, that is, $OPT(S)$ is a
shortest possible list of operations that realizes $S$. An algorithm $\A$
(online or offline) is {\em $f(n)$-competitive} if we have $\A(S) =
O(f(n)\cdot OPT(S))$ for all sequences~$S$. It is {\em dynamically
optimal} if it is $O(1)$-competitive.

\paragraph{Prior works}
The study of self-adjusting BSTs to minimize the overall cost over a
sequence of accesses was initiated by Allen and Munro \cite{allen}
with their analysis of the move-to-root and the simple exchange
heuristics, and then by Sleator and Tarjan with the introduction of
Splay trees \cite{splay}, which they conjectured to be dynamically
optimal.  They show how the running time of Splay trees can be upper
bounded in several ways as a function of the access sequence. They
prove the \emph{balance theorem} (accesses run in $O(\log n)$
amortized), the \emph{static optimality theorem} (any sequence of
accesses runs within a constant factor of the time to run it on the
best possible static tree for that sequence; in particular it reaches
the entropy bound), the static finger theorem (access $x$ runs
in $O(\log d(x,f))$, where $d(x,f)$ is the number of keys between the
query item $x$ and any fixed \emph{finger} element $f$), the
\emph{working set theorem} (access $x$ runs in time $O(\log w(x))$
where $w(x)$ is the number of distinct elements accessed since the
previous access to $x$), and the \emph{scanning theorem} (accessing
all nodes in symmetric order takes time $O(n)$). They also conjectured
the \emph{dynamic finger theorem} (access to $y$ runs in amortized
$O(\log d(x,y))$ where $x$ is the previous item in the access
sequence), which was subsequently proved by Cole \cite{cole,cole2}.
All bounds above are amortized.

On another front, Wilber \cite{wilber89} gave a formal analysis of
several variants of the BST model, providing equivalence reductions
between them, and provided two lower bounds on the number of
operations that any BST algorithm must perform for a given sequence.
In particular, he proved that the bit reversal sequence requires
$\Omega(\log n)$ amortized operations per access. 
These lower bounds were recently generalized
in \cite{rectcover,BST_SODA2009}. 
Splay trees were also shown to be \emph{key independent optimal}
\cite{keyindependent}, that is, they are $O(1)$-competitive if the order of the
keys is arbitrary or random, and that they are $O(1)$-competitive with
respect to a wide class of balanced BST algorithms \cite{Georgakopoulos200464}.

New bounds have been designed: the \emph{queueish} bound (opposite of the
working set bound: the number of elements \emph{not} accessed since
the last access to $x$) was shown not to be achievable by any BST algorithm
\cite{queaps}. Recent papers have attempted to engineer a BST that satisfies the
unified property, a bound that implies both the dynamic finger and the working set
bound \cite{unified,unified2}. The \emph{skip-splay trees} \cite{skipsplay}
perform each access within a multiplicative factor $O(\log \log n)$ of
the unified bound, amortized. The \emph{layered working set trees} \cite{workingsettrees}
are BSTs that achieve the working set bound worst case. By combining
it with the skip-splay structure, the authors show how to achieve the
unified bound, amortized, with an additive cost of $O(\log \log n)$. 

The first significant breakthrough on the competitive analysis of BST
algorithms came with the invention of \emph{tango trees} \cite{tangotrees}, the first
provably $O(\log\log n)$-competitive BST. This result was subsequently
improved independently by the \emph{multi-splay trees}
\cite{multisplay} and the \emph{chain-splay trees} \cite{chain-splay} which both offer
the additional guarantee of performing each access in $O(\log n)$
amortized time. Further properties of multi-splay trees were proved in
\cite{multisplayprop}, where they were shown to satisfy static
optimality, the static finger property, the working set property, and
key-independent optimality. They further satisfy the dequeue property
which is not known to be satisfied by splay trees.

In recent years, the question was raised as to whether the good
amortized properties could be reconciled with the $O(\log n)$ worst
case bounds satisfied by well balanced trees such as AVL or red-black
trees.
Such results were known for static trees \cite{balancestatic}, however
recent works gave indication that strong balance constraints at every
node forces the working set bound to be an amortized lower bound, thus forbidding
any such tree to have stronger properties such as the dynamic finger property
\cite{skiplists_SODA08} (the proof was given for self-adjusting
skip-lists and B-trees, however the proofs can easily be adapted to
BST with balance constraints at every node). However, it remained open
whether relaxing the balance condition to just bounding the height of
the tree would be compatible with obtaining better amortized performances.
In \cite{competworstcase}, a BST based on Tango trees \cite{tangotrees} is engineered to be both
$O(\log\log n)$-competitive and guarantees $O(\log n)$ worst case
access time for each access. However, this structure is unlikely to
possess all the other desirable properties of Splay trees.

\paragraph{Our results}

In this article we show that it is possible to automatically transform \emph{any}
BST algorithm into one that provides worst case time guarantees per access
while keeping the same asymptotic amortized running times. Our core result
shows how to keep a BST balanced while losing only a constant factor in the
running time:
\begin{itemize}
\item Any BST algorithm $\A$ on tree $T$ can be transformed into a BST
algorithm $\A'$ on a tree $T'$ whose amortized cost is within a constant
factor of the original algorithm, and for which the depth of $T'$ is always
$O(\log n)$. If $\A$ is online, so is $\A'$.
\end{itemize}
Using this, we then show how to de-amortize the BST and answer each query
in $O(\log n)$ worst case cost:
\begin{itemize}
\item Any BST algorithm $\A$ on tree $T$ can be transformed into a BST
algorithm $\A''$ on tree $T''$ such that for any access sequence $S$,
$|\A''(S)| = O(|\A(S)|)$ and each access to a node is performed in $O(\log
n)$ operations worst case. If $\A$ is online, so is $\A''$.
\end{itemize}

Finally, we show that we can extend the method to maintain the additional
features of real-world online BST algorithms described above, in a way
which turns amortized upper bounds on the BST algorithm into worst case
performance per access. In particular, we have:

\begin{itemize}
\item Any online BST algorithm $\A$ on tree $T$ that performs $k$ accesses
in $O((n+k)\log n)$ operations can be transformed into an online BST
algorithm $\A'''$ on tree $T'''$ such that for any access sequence $S$,
$|\A'''(S)| = O(|\A(S)|)$ and each access to a node is performed in $O(\log
n)$ operations worst case.
If $\A$ works in the pointer machine model, with working space being $O(1)$
words of information in the nodes and $O(1)$ global working variables, and
computes each access to a key in time proportional to the number of BST
operation of the access, then so does $\A'''$.

\end{itemize}

Applying this transformation to Splay trees, we obtain a BST that
executes every sequence within a constant factor of the Splay tree and
thus satisfies the scanning theorem, the working set property, static
optimality, the key-independent optimality, the static finger
property, the dynamic finger property, and that performs each access
in $O(\log n)$ worst case. 
By transforming Multi-Splay Trees, we obtain a BST that
is $O(\log \log n)$ competitive, satisfies the scanning theorem, the
working set property, static optimality, the key-independent
optimality, the static finger property, and
performs each search in $O(\log n)$ worst case time. 
Furthermore, if there is a dynamically optimal BST algorithm, then
there is one that additionally performs every search in $O(\log n)$
operations worst case.

\paragraph{Structure of paper}

In the next section we show how to implement a stack as a binary
search tree with bounded height. This will be used as a building block
in Section~\ref{sec:simulation} to simulate the operations of any BST
using another BST of bounded height. Finally, we show how to use this
rebalanced tree to de-amortize the BST algorithm in Section~\ref{sec:de-amortization}.

\section{Pop-Tarts}

We start by implementing a stack using a balanced BST. We differentiate
internal nodes, which always have two children, and leaves which have no
children (leaves can also be seen as empty pointers). In order to fit
the stack data structure in the BST model, we assume that nodes to be
pushed onto the stack appear as the parent of the root of the current
stack, and that nodes are pushed onto the stack in decreasing key order
(that is, after the push operation the old stack is the right child of
the newly inserted node, and its left child is a leaf). Our later
application of the stack structure fulfils these assumptions. An empty
stack is composed of one leaf.  The structure will maintain the
invariant that the left child of the root is always a leaf, to allow for
easy pop operations. After each push or pop operation, the structure is
allowed to perform a sequence of operations in the BST model (finger
movements and rotations), and at the end of the sequence, the finger is
back at the root. Leaves can have a weight associated to them, and we
use the convention that internal nodes all have weight 1 (it would not
be difficult to generalize these structures to support arbitrary
internal weights, however this is not necessary for our application).

A BST implementing a stack in this manner we call a
\emph{Pop-tart}\footnote{Pop-Tarts are a line of crazy good~\cite{crazygood} breakfast products that {\em pop} out of the toaster, which remind us of popping a stack. Pop-tart is a trademark of the Kellogg
Company.}. A pop-tart is \emph{good} if push and pop operations are
performed in $O(1)$ amortized time and $O(\log n)$ worst-case time.  It
is \emph{crazy good}~\cite{crazygood} if it is good and the depth of every leaf of weight
$w$ is $O(\log (W/w))$, where $W$ is the
total weight of all leaves in the pop-tart, or $O(\log n)$ for an
unweighted pop-tart with $n$ leaves\footnote{We slightly abuse the big-Oh notation and write
$O(\log(W/w))$ to mean a function which is smaller than $c\log(W/w)+d$
for some constants $c$ and $d$.}.

In the remainder of this section, we will describe three pop-tart
structures. The first two lay down ground concepts that will be used to
construct the third pop-tart (Chocolate), which is always crazy good.

\paragraph{Vanilla Pop-Tart.}
Implementing a good pop-tart is easy. In fact, performing no BST
operations after each push or pop operation will produce a linear tree
with exactly $O(1)$ time per operation. This elementary implementation
is called \emph{Vanilla Pop-Tart}.  A vanilla pop-tart will be
crazy-good if the weight of each pushed leaf is always larger than the
total weight of all other leaves in the pop-tart.

\begin{lemma}
The Vanilla Pop-Tart is crazy good if nodes are added in decreasing key
order and new leaves have weight larger or equal to the total weight of
all other leaves in the pop-tart.  That is, it uses $O(1)$ time per push
and pop operation and the depth of a leaf of weight $w$ is at most
$1+\log W/w$ where $W$ is the total weight of all leaves in the
pop-tart.
\end{lemma}
\begin{proof}
The proof is by induction. If the pop-tart contains one leaf, then it is
at depth 0, this covers the base case.  Assume by induction that the
lemma is true for the right subtree of the root, which is of total
weight $W'$. Then the left child of the root is the last added leaf and
it has weight at least $W'$, thus, $W \geq 2W'$. The left child of the
root is at depth $1 \leq 1+\log{W/w}$. Any other leaf in the tree by
induction is at depth at most $2+\log{W'/w} \leq 1+\log{W/w}$.
\end{proof}

\paragraph{Cherry Pop-Tart.}
We now describe the \emph{Cherry Pop-Tart}, which is a crazy good
pop-tart if all leaves have weight 1. Although Cherry Pop-tarts are not
used explicitly in this paper, they serve as a warm up, introducing some
key concepts needed to define the Chocolate Pop-tart structure, which is
used later.

The algorithm used is a variant of a 2-4 tree implemented as a BST.  On
a high level, it may be viewed as reversing edges on the leftmost path
in a red-black tree, and then having a permanent finger at the leftmost
internal node (effectively making it the root of the BST).

In greater detail: The Cherry Pop-tart is a BST with the nodes on the
right path of the tree grouped into layers. A layer consists of
consecutive nodes on the right path, and the left subtrees of these
nodes are called \emph{crumbs}. The right child of the last node in the
layer is the top node of the next layer (except for the last layer,
where it is the original leaf of the initial empty stack). By definition
of BSTs, the layers are linearly ordered, that is, all keys in a layer
are smaller than the keys in the next layer.

We number the layers as follows: the layer containing the root is
layer~0, the next one along the right path is layer~1, and so on.  We
maintain the invariants that each layer has between 1 and 3 nodes on the
right path (hence that many crumbs), and that the crumbs pointed to by
layer~$i$ (called \emph{$i$-crumbs}) are perfectly balanced trees
containing exactly $2^i$ leaves. See Figure~\ref{fig:cherry-pop-tart}.
\begin{figure}
\begin{center}
\includegraphics[scale=0.50]{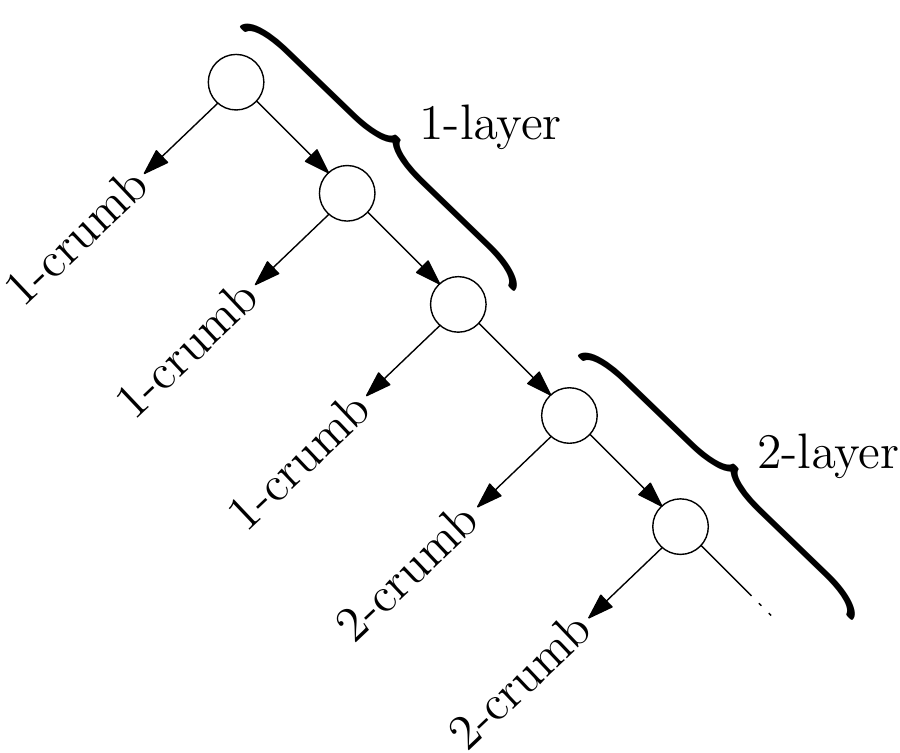} 
\end{center}
\caption{layers and crumbs of a Cherry Pop-tart.}
\label{fig:cherry-pop-tart}
\end{figure}

The invariant is true for a pop-tart containing one node: that node is
layer~0 and it points to one $0$-crumb (containing one leaf). When a new
node is pushed as the parent of the root, it is added to
layer~0. Layer~0 therefore has one more node and one more
$0$-crumb. Either the new layer~0 still has no more than 3 crumbs,
maintaining the invariant, or layer~0 now has 4 $0$-crumbs (each
composed of exactly one leaf). In this case, we perform a left rotation
between the last two nodes of the layer. This replaces the last two
nodes of the layer with one node whose left pointer points to a
$1$-crumb. We now move that node from layer~0 to layer~1. See
Figure~\ref{fig:cherry-pop-tart-rot}.
\begin{figure}
\begin{center}
\includegraphics[scale=0.50]{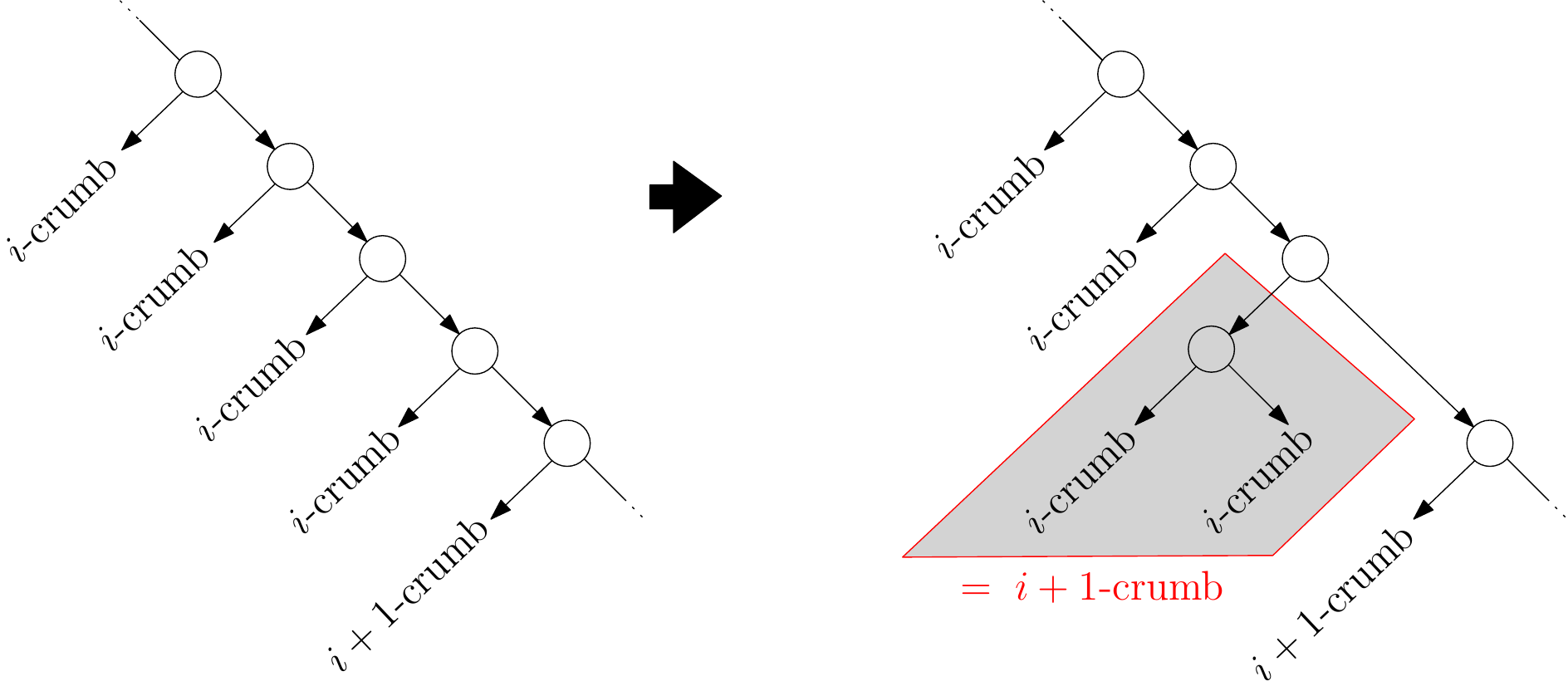} 
\end{center}
\caption{Restoring the Cherry Pop-tart invariant at level~$i$.}
\label{fig:cherry-pop-tart-rot}
\end{figure} 
Again, the reconfiguration could either stop there or ripple down
further.  In general, as a node is added as the parent of the first node
in layer $i$, either layer $i$ still has no more than 3
$i$-crumbs, or we preform a rotation on the node between the last two
crumbs, forming a $(i+1)$-crumb with twice as many leaves which is
inserted into layer $i+1$. A pop operation works symmetrically, by
removing the first node of layer~0 (whose left child is a leaf) and
restoring the invariant, that is, if layer~0 contains no more nodes,
we perform a right rotation on the first node of layer~1, transforming
it into two nodes that are moved into layer~0.  If layer~1 is now empty,
we repeat the operation on the first node of layer~2 and so on.

\begin{lemma}
The Cherry Pop-Tart is crazy good if nodes are added in decreasing
key order and all leaves have weight~1. That
is, it uses $O(1)$ amortized time and $O(\log n)$ worst case time per
push and pop operation and its tree has height $O(\log n)$. 
\end{lemma}
\begin{proof}
To show that a push or pop operation has amortized cost
$O(1)$, we assign a potential of 0 to layers with 2 nodes,
and a potential of 1 to layers with 1 or 3 nodes. A push or pop
operation has actual cost proportional to the number of layers that
had to be readjusted to restore the invariant. Each readjusted layer had
a potential of 1 before the operation (i.e., had 3 nodes before a push
or 1 node before a pop) and of 0 after the operation (i.e., has 2
nodes exactly). Therefore, the decrease of potential pays exactly for
the readjustments. The insertion or deletion in the last layer possibly
increases its potential by 1, which is the amortized cost of the
operation. Therefore, this pop-tart is good.

Since layer $i$ has at least one $i$-crumb containing $2^i$ leaves, a
pop-tart with $n$ leaves has at most $\log n$ layers, each having crumbs
of height $O(\log n)$, thus the total height of the tree is $O(\log n)$.
So in the unweighted case, this pop-tart is crazy-good.
\end{proof}

The next lemma shows that the Cherry Pop-tart is crazy-good even in
some weighted cases. 
\begin{lemma}\label{lem:cherry-pop-tart}
The Cherry Pop-Tart is crazy good if nodes are added in decreasing
key order and new leaves are added with increasing weights.
That is, it uses $O(1)$ amortized time per push and pop operation and
the depth of a leaf of weight $w$ is $O(\log W/w)$ where $W$ is the
total weight of all leaves in the pop-tart.
\end{lemma}
\begin{proof}
We use the exact same structure as in the previous lemma. Observe that
by the conditions in the lemma, an inorder traversal of the tree will
meet the leaves in order of decreasing weight. Since $i$-crumbs
contain~$2^i$ leaves, the layer containing the $k^{th}$ heaviest leaf in
the pop-tart has index at most $\log k$, hence has crumbs of depth at
most $\log k$. So the depth of the $k^{th}$ heaviest leaf is at most
$4\log k$. If the $k^{th}$ heaviest leaf is of weight $w$, then the
total weight $W$ of all leaves in the pop-tart is at least $kw$, hence
the depth of that leaf is at most $4 \log k \leq 4\log {W/w}$. Thus, the
pop-tart is crazy good.
\end{proof}

In order to allow for arbitrary weight order, we will have to modify
slightly the data structure. We call the next structure the
\emph{Chocolate Pop-Tart}.

\paragraph{Chocolate Pop-Tart.}
Again, the structure will be decomposed into a sequence of layers whose
nodes form a right path and point to crumbs. This time, the right path
of the $i^{th}$ layer will be composed of 1 to 3 \emph{regular} nodes
whose left child is an $i$-crumb, then a \emph{next} node whose left
child points to the next layer and whose right child points to a subtree
called the \emph{icing}. This will be called the \emph{structural invariant}.  See Figure~\ref{fig:chocoloate-pop-tart}.
\begin{figure}
\begin{center}
\includegraphics[scale=1]{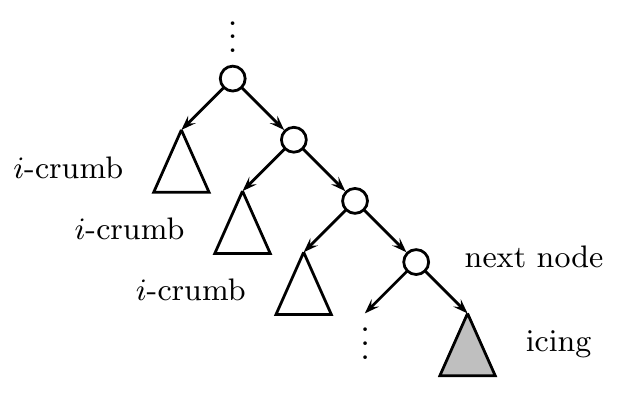} 
\end{center}
\caption{Level $i$ in the Chocolate Pop-tart.}
\label{fig:chocoloate-pop-tart}
\end{figure} 
The icing is itself a stack, implemented using a Vanilla Pop-tart (that
is, a simple linear tree), whose leaves will be
\emph{frozen}\footnote{or \emph{frosted}} subtrees of the chocolate
pop-tart. In order for the icing to be crazy-good, we will ensure that
the nodes (frosted subtrees) pushed onto it will always be at least as
heavy as the total weight of the icing.  The subtrees to be frosted and
pushed into the icing of level~$i$ will always be the next node and the
entire subtree rooted at the top node of level $i+1$.  Therefore, we
maintain the invariant that the total weight of layer $i+1$ (that is,
the the total weight of the subtree rooted at the topmost node of that
layer) is smaller than the total weight of the icing of layer $i$
(\emph{thick icing invariant}). If
violated, layer $i+1$ will be frosted and pushed into the icing, to
maintain the invariant.

The last layer, say, layer $i$, is incomplete: it is composed of 0 to 3
regular nodes, has no pointer to the next layer, and always contains an
icing as its rightmost subtree. It can only have 0 regular nodes if the
icing contains exactly one element (which is always an $i$-crumb).

As before, when a new node is pushed onto the $i$-layer (starting with
$i=0$), 
either the $i$-layer has at most 3 regular nodes, in which case we are
done, or it contains 4 regular nodes and we need to restore the structural
invariant.  We start by performing a left rotation between the two
lowest regular nodes in the layer, creating an $(i+1)$-crumb.  We have
two cases to consider. If the $i^{th}$ layer is not the last one, then
we perform a left rotation between the next node and the lowest regular
node, to move the new $(i+1)$-crumb and its node to the $(i+1)^{th}$
layer.  On the other hand, if the $i^{th}$ layer is the last one, then
it has no next node. Then the lowest regular node becomes a next node
which points to the new $(i+1)^{th}$ layer. That $(i+1)^{th}$ layer
contains 0 regular nodes, no next node and an icing which contains the
$(i+1)$-crumb as its only leaf.

Having done this, there are again two cases to consider: if the total
weight of the subtree rooted at the (new) top node of the $(i+1)^{th}$
layer is smaller than the total weight of the icing of the $i^{th}$
layer, then we proceed with the insertion of the $(i+1)$-crumb, by
restoring the structural invariant if necessary, and so on. Otherwise, we restore
the thick icing invariant by frosting the
$(i+1)^{th}$ layer without modifying it further (even if it contains now
4 regular nodes), and push it and its parent node (the next node of the
$i^{th}$ layer) into the icing of the $i^{th}$ layer. The $i^{th}$ layer
then becomes the last layer. It has no next node and two regular nodes.

The deletion operation is symmetric: when the first regular node of the
$i^{th}$ layer is deleted, either the layer still has at least one
regular node left, in which case we are done, or we have to restore the
structural invariant.  If $i$ is not the last layer, we pull two nodes and their
associated $i$-crumbs from the $(i+1)^{th}$ layer (by performing two
right rotations and possibly recursively restoring the invariant in the
$(i+1)^{th}$ layer).  If the $(i+1)^{th}$ layer is only composed of an
icing (which then contains one frosted $(i+1)$-crumb), we defrost the
icing, perform a right rotation, transforming the next layer into two
regular nodes pointing to $i$-crumbs and the $i^{th}$ layer becomes the
last one.  On the other hand, if $i$ is the last layer, then we pop a
frosted subtree from the icing (unless it contains only one leaf), and
perform a right rotation to turn the frosted subtree into one regular
node and a next node, the latter pointing to the new, unfrosted,
$(i+1)^{th}$ layer and to the remaining icing.

\begin{lemma}\label{lem:chocolate-pop-tart}
The Chocolate Pop-Tart is crazy good if nodes are added in decreasing
key order and new leaves are added with arbitrary weights.
That is, it uses $O(1)$ amortized time per push and pop operation and
the depth of a leaf of weight $w$ is $O(\log W/w)$ where $W$ is the
total weight of all leaves in the pop-tart.
\end{lemma}
\begin{proof}
We first show that the Chocolate Pop-tart is good, that is, it uses
$O(1)$ amortized time per push and pop operation.
For this, we assign a potential of 0 to layers with 2 regular nodes,
and a potential of 1 to all other layers. 
A push operation will cause a bunch of reconfigurations in successive
layers, that end in either adding a crumb to a layer that does not
overflow, or pushing an element in the icing of a layer. Either case
costs $O(1)$ amortized. As in the case of Cherry Pop-tarts, it is
easily verified that every layer that overflows had 3 regular nodes
before, and thus a potential of 1, and two regular nodes after, so a
potential of 0 (except possibly for the last rearranged layer).
Likewise, during a pop operation, the potential of a rearranged layer
(except the last one) goes from 1 to 0 since the number of regular
nodes it contains goes from 1 to 2.
Thus, the decrease of potential of a layer during a push or a pop pays
for its rearrangement, while the amortized cost of $O(1)$ pays for the
potential increase and the rearrangement in the last node and the push
in the icing if it occurs.

It now remains to prove that the depth of a node of weight $w$ is
$O(\log W/w)$. 
The proof will be by induction on the layer number. Consider the subtree
rooted at the first node of the $i^{th}$ layer and let $W_i$ be the
total weight of that subtree.
Assume by induction that at any moment in the algorithm, any leaf of weight $w$ has
depth $i+6+7\log{W_i/w}$ starting from the
root of the $i^{th}$ layer. 
We want to show that in the subtree rooted at the first node of the
$(i-1)^{th}$ layer, any leaf of weight $w$ has depth 
$(i-1)+6+7\log{W_{i-1}/w}$.
Obviously, the hypothesis is true for an $i^{th}$ layer 
that contains only an icing with one frosted $i$-crumb,
since all its leaves are at distance $i$;
this covers the base case.

For a $(i-1)^{th}$ layer, we consider the leaves located 
(i) in $(i-1)$-crumbs pointed by regular nodes, 
(ii) in the $i^{th}$ layer if it exists, and 
(iii) in the icing of the $(i-1)^{th}$ layer.
Any leaf of type (i) is at distance $\leq 3+i-1$ which is small enough.
For type (ii) leaves, notice that as long as $i$-crumbs are being moved
from the $(i-1)^{th}$ layer to the $i^{th}$ layer without being
frosted and pushed to the icing, $W_i \leq W_{i-1}/2$. Therefore, for
any leaf of weight $w$ in the subtree of the $i^{th}$ layer, the
depth of that leaf is at most 
$$4 + i + 6 +7\log {W_i/w} \leq 10+i+7\log {W_{i-1}/w} - 7
\leq (i-1) + 4 + 7\log {W_{i-1}/w}$$
which is below the desired bound.

Finally for case (iii), since the icing of the $(i-1)^{th}$ layer is
implemented as a Vanilla pop-tart and the frosted subtrees are pushed
with (total) weights always larger than all other leaves (frosted
subtrees) in the icing, the icing is crazy good, that is, a frosted
subtree of total weight $W$ will have its root at depth at most 
$5+\log {W_{i-1}/W}$. 
Let $p$ be the parent of the frosted subtree containing the node of
weight $w$, let $W_p$ be the
weight of the subtree rooted at $p$. The depth of $p$ is at most 
$4+\log {W_{i-1}/W_p}$ since the left child of every node on the right
path of the icing contains at least half of the weight of that node.  
Every frosted subtree has its first node whose left pointer points to
a possibly heavy $i$-crumb, and whose right pointer points to what
used to be the $i^{th}$ layer at some point in time. Let $W'$ be the
weight of that $i^{th}$ layer. Then $W' \leq W_p/2$ otherwise the
$i^{th}$ layer would have been frosted earlier. 
By induction, a leaf of weight $w$ in this former $i^{th}$ layer must
have depth no more than 
$$4+\log {W_{i-1}/W_p} + 2 + i + 6 + 7\log {W'/w} $$
$$\leq 12 + i + \log {W_{i-1}/W_p} + 7\log {W_p/w} - 7$$
$$\leq (i-1) + 6 + 7\log{W_{i-1}/w}$$
which is the desired bound. 
A leaf in the $i$-crumb pointed by the left pointer of the root node
of the frosted subtree has weight at most $W_p$, and its depth is
$$4+\log {W_{i-1}/W_p} + 2+ i 
\leq   (i-1) + 6 + 7\log{W_{i-1}/w}.$$
This completes the induction proof. For $i=0$, we have that any leaf
of weight $w$ has depth at most $6+7\log{W/w}$, so the chocolate
pop-tart is crazy-good for arbitrary weights. 
\end{proof}

Note that all pop-tarts described in this section can also be flipped
to maintain elements pushed in increasing order. If the cherry or
chocolate pop-tarts need to be implemented in a real-world BST,
$O(1)$ extra bits of information in each node is sufficient
for storing the function of that node (regular, next, icing, crumb). 

\section{Simulation}\label{sec:simulation}
We now show how to efficiently simulate any BST algorithm while keeping
the tree of logarithmic height. The method will work for trees with
weighted nodes as well. Let $w_i$ be the weight of the node with key $i$
and let $W=\sum_{i=1}^n w_i$.  For unweighted trees, set $w_i=1$ and
$W=n$.  We represent the tree $T$ of the original BST algorithm using a
heavy path decomposition. To construct this decomposition, we denote
every edge of $T$ as either \emph{solid} or \emph{dotted}. For each
non-leaf node, the edge to its child with largest total subtree weight
(or the left child, in case of a tie) is a solid edge, and the edge to
its other child is dotted. The solid edges form \emph{heavy paths}
connected together by dotted edges.

We simulate the original BST algorithm as follows: When its finger is at
the root of $T$, each heavy path is implemented using a pair of weighted
pop-tarts: a heavy path from node $y$ to node $x$ (with $y$ an ancestor
of $x$) is a sequence of nodes that can be decomposed into the
subsequence $L(y,x)$ of nodes smaller than $x$ on the path, and the
subsequence $R(y,x)$ of nodes larger than $x$ on the path. Note that
$L(y,x)$ is increasing, and $R(y,x)$ is decreasing. In our simulation,
the end of the path $x$ does not change, but $y$ can move up or down
along the path to the root. As $y$ moves up, the new nodes are added to
$L(y,x)$ in decreasing order, or to $R(y,x)$ in increasing order.

The sequences $L(y,x)$ and $R(y,x)$ will each be stored in the weighted
chocolate pop-tart structure described in the previous section, and
these two pop-tarts will be left and right children of $x$,
respectively, see Fig.~\ref{fig:double-pop-tart}.
\begin{figure}
\begin{center}
\includegraphics[scale=0.50]{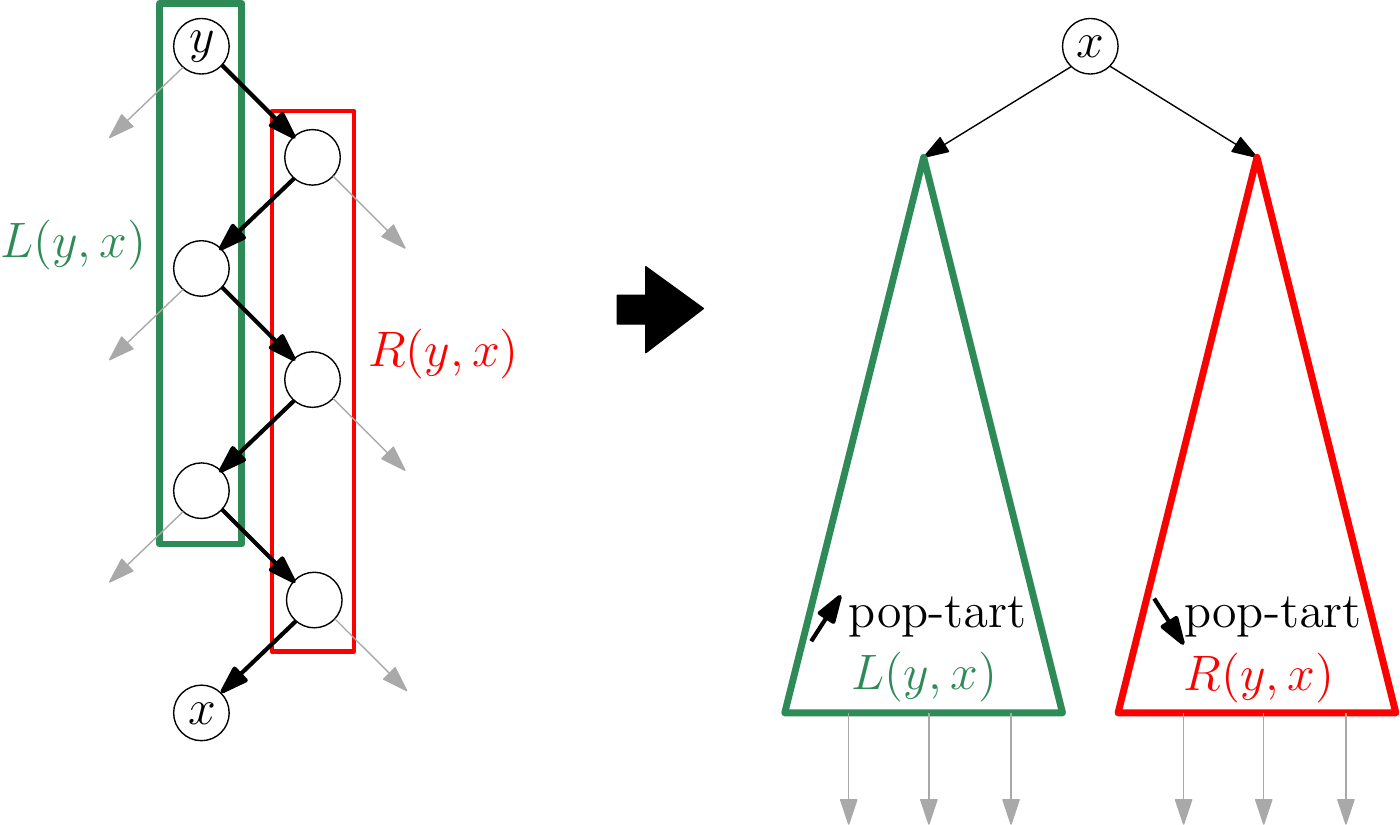}
\end{center}
\caption{Representing a heavy path with Pop-tarts.}
\label{fig:double-pop-tart}
\end{figure}
Each node on the path is connected via a dotted edge to a subtree which
will be considered as a leaf in the pop-tart, whose weight is exactly
the total weight of all the nodes in that subtree. The subtrees
contained in those leaves will be structured in the same manner,
recursively.  The nodes in the tree will contain two extra bits, one to
determine if the edge to its parent node is solid or dotted, and another
to determine if the next node on its heavy path is in $L(y,x)$ or
$R(y,x)$.

When the finger $f$ is not at the root $r$ of the tree, the path from
the finger to the the root is also represented as a pair of pop-tarts in
a similar way, but this time upside-down (see Fig.~\ref{fig:finger}).
\begin{figure}
\begin{center}
\includegraphics[scale=0.65]{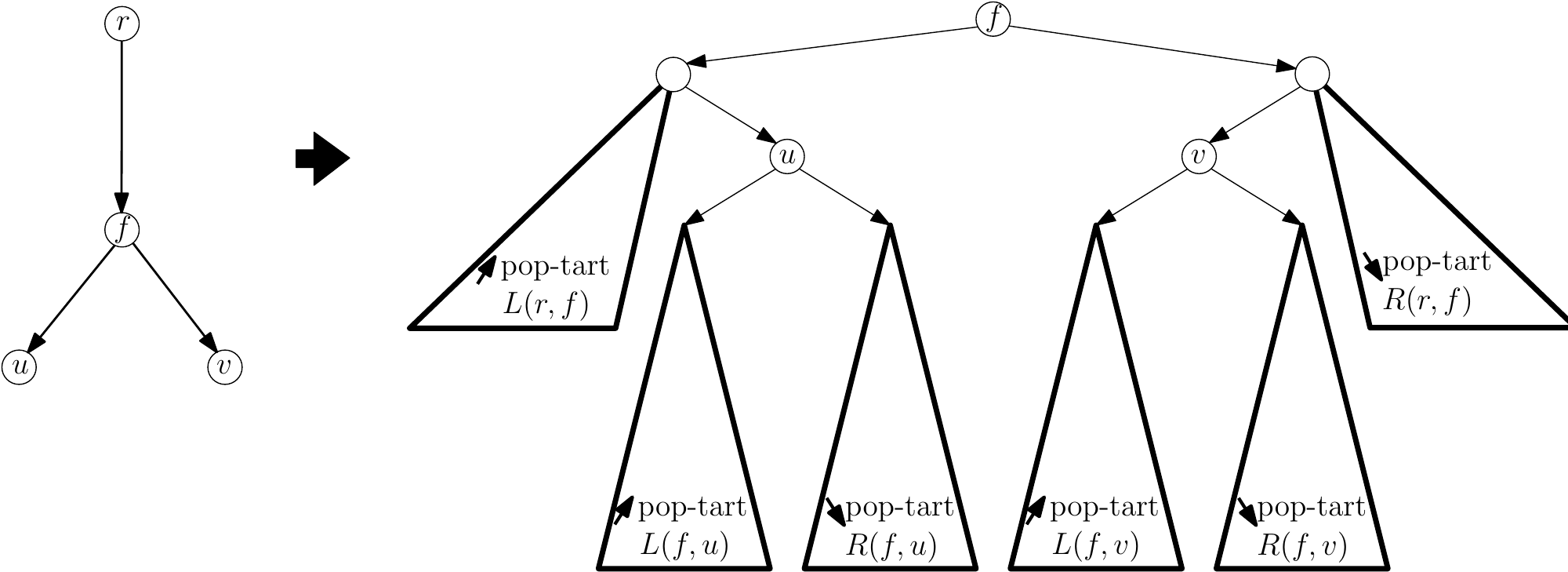}
\end{center}
\caption{Representing the finger in general position.}
\label{fig:finger}
\end{figure}
Thus, as $f$ walks down, the elements of $L(r,f)$ are added in
increasing order, and the elements of $R(r,f)$ are added in decreasing
order.  Hence, finger movements in the original BST algorithm can be
implemented using one push and one pop operation by transferring a node
from one pop-tart to the other using $O(1)$ rotations. Likewise,
rotations in the original BST algorithm only involve the first few nodes
on the pop-tarts linked from the finger, and thus can be implemented in
$O(1)$ rotations and push/pop operations. Note that the finger in the
tree maintained by our simulation always stays at the root.

Any path from the root to a node $x$ of weight $w$ uses at most $\log{W/w}$
dotted edges. Further, let $W_1, W_2, \ldots, W_k$ be the total weights of
the successive heavy paths (along with their descendants) on the path from
the root to $x$. By Lemma~\ref{lem:chocolate-pop-tart}, the $i^{th}$ heavy
path will be stored at depth $O(\log(W_{i-1}/W_i))$ in the pop-tart of the
$(i-1)^{th}$ heavy path, and node $x$ will be at depth $O(\log(W_{k}/w))$
in the pop-tart of the last heavy path.  Thus, the total depth of $x$ in
the tree is bounded by a telescoping sum that sums up to $O(\log(W/w))$.
Clearly, if $\A$ is online, so is $\A'$. We obtain:

\begin{theorem}\label{thm:simulation}
Given a BST algorithm $\A$ with a starting tree $T$, there is a BST
algorithm $\A'$ with a starting tree $T'$ such that $|\A'(S)| =
O(|\A(S)|)$,
and such that the depth of a node $i$ in $T'$ is always $O(\log (W/w_i))$
and the finger is always at the root of $T'$. If $\A$ is online, so is
$\A'$.
\end{theorem}

If the simulation needs to be implemented in a real-world BST, $O(1)$ extra
bits per node is sufficient for storing the structure of
the original tree and the function of each node in the simulation:
each node needs to indicate wether each of its 
children is part of the same heavy path or not, and for all nodes on
the path from $r$ to $f$, a bit must be stored to determine if the
next node on the path is stored in $L(r,f)$ or in $R(r,f)$.
Note that if $T$ is unbalanced, it is necessary to restructure it into
$T'$ in order to obtain the depth bound. However if the starting
position of $T$ already have this property, we can start with
$T$ unchanged and restructure the tree
during the execution of the algorithm every time the finger enters a
yet unexplored subtree.

\section{De-amortization}\label{sec:de-amortization}
We are now ready to show how to de-amortize BST algorithms.

\begin{theorem}\label{thm:offline}
For any BST algorithm $\A$ with a starting tree $T$ there is a BST
algorithm $\A''$ with a starting tree $T''$ such that for any access
sequence $S$, $|\A''(S)| = O(|\A(S)|)$ and each access to a node is
performed in $O(\log n)$ operations worst case. If $\A$ is online, so is $\A''$.
\end{theorem}
\begin{proof}
Using Theorem~\ref{thm:simulation}, transform $\A$ and $T$ into $\A'$ and
$T'$ such that the depth of node $i$ in $T'$ is always $c\log n$ for some
constant $c$.  Algorithm $\A'$ is then modified in the following way: while
running the sequence of operations in $\A'(S)$, every time $c\log n$
operations from the original $\A'(S)$ sequence have been performed without accessing the next unaccessed
element of the input sequence, access this element by moving the finger to it
and back (thereby inserting $\leq 2c\log n$ extra BST operations into the sequence at this point).
Thus every access is performed in worst case $3c\log n$, and the total cost of the sequence is
the same within a factor $3$. If $\A$ is online, so is $\A''$.
\end{proof}

Again, it is usually necessary to transform the starting tree in order to
achieve a $O(\log n)$ worst case bound per access, for example in the case
when $T$ is very unbalanced.  If the starting tree $T$ has height $O(\log
n)$ however, it is not necessary to modify it and we can, as described
above, restructure the tree during the execution of the algorithm every
time the finger enters a yet unexplored subtree.

As mentioned in the introduction, real-world BST algorithms normally work
in the pointer machine model, with working space for the algorithm being
$O(1)$ words of information in the nodes of the tree, and $O(1)$ global
working variables. Additionally, they can be implemented to find their BST
operations for a key~$s_i$ in time proportional to the number of these
operations (i.e., in time proportional to the cost in the BST model).

A natural goal is that our de-amortized output algorithm should adhere to
these constraints if the input algorithm does. We now describe how to do
this, given bounds on the amortized behaviour of the input BST
algorithm. In particular, we show how any real-world online BST algorithm
with $O(\log n)$ amortized time bounds (such as e.g.\ Splay Trees) can be
transformed into an online BST algorithm with $O(\log n)$ worst case time
bounds, while not changing their running time on any sequence by more than
a constant factor. In the following theorem, the formulation is slightly
more general.

\begin{theorem}\label{thm:online}
Let $f(n)$ be a function in $\Omega(\log n)$. For any online BST algorithm
$\A$ that for some starting tree $T$ is guaranteed to perform $k$ accesses
in $O(n f(n) + k f(n))$ operations, there is an online BST algorithm
$\A'''$ and a starting tree $T'''$ such that $|\A'''(S)| = O(|\A(S)|)$ for
any access sequence $S$, and such that $\A'''$ performs each access to a
node in $O(f(n))$ operations worst case.
If $\A$ works in the pointer machine model, with working space being $O(1)$
words of information in the nodes and $O(1)$ global working variables, and
computes each access to a key in time proportional to the number of BST
operation of the access, then so does $\A'''$.
\end{theorem}
\begin{proof}
The general idea is the same as in Theorem~\ref{thm:offline}, except
that $\log n$ is replaced by $f(n)$.
The main problem to overcome is that for some $s_i$'s, the number of BST
operations of $\A'$ may be larger than $f(n)$, due to the amortization
in~$\A$ (and the amortization added when transforming it into~$\A'$). These
BST operations cannot all be executed before the access to $s_i$ has to be
finished by a traversal of the balanced tree of~$\A'$ and $s_{i+1}$ has to
be served next by the online algorithm~$\A'''$. In short, in the execution
of~$\A'$, its point in the input sequence can lag behind that of~$\A'''$,
and the problem is how~$\A'''$ efficiently can keep track of what
operations to do next when executing~$\A'$.

We do this by maintaining a queue~$Q$ containing the keys whose accesses
have already been performed in $\A'''$ but whose BST operations in the
execution of $\A'$ still have to be done. Thus, $Q$ always contains a
(possibly empty) suffix of the keys $s_1,s_2,\dots,s_i$, where $s_i$ is the
key last accessed by~$\A'''$. The $\A'$ operations of the oldest key in~$Q$
may be partly executed, and we store the state of the process of $\A'$ on
that key in the global variables of $\A'''$ (we assume such a state can be
stored in $O(1)$ words for~$\A$, which implies that it also can be done
for~$\A'$).

To adhere to our notion of practical BST algorithms, we implement the
queue by a linked list of queue nodes, with each node of the list stored in a node of
the tree as a pair of words. The first word is the key stored at that
position in the queue and the second word is a list pointer to the next
queue node, represented by the key of the tree node containing
that queue node.  Note that following a list pointer may require walking
$O(\log n)$ steps in the tree and so enqueue or dequeue operations will
cost that much.

We now give the details of~$\A'''$. It uses the following three basic
routines.
\textbf{A:} Restart the $\A'$ process of the last key in~$Q$, and then
perform $\A'$ process work on this and the following keys~$Q$, doing a
dequeue each time the process for the last key finishes. The routine
ends when $\A'''$ has done $d \cdot f(n)$ BST operations, or $Q$ runs
empty. Here, $d$ is a constant to be determined later.
\textbf{B:} Perform $\A'$ BST operations on the newest key of $\A'''$. The
routine ends when $f(n)$ such operations have been performed, or the
operations have all been done.
\textbf{C:} Access the newest key of $\A'''$ by a search in the tree
maintained. Enqueue the key in~$Q$.

Given these routines, the actions of~$\A'''$ on the next input key are:
\begin{quote}
\textbf{IF} $Q$ is not empty:\\
\mbox{}~~~~do \textbf{A}\\
\textbf{IF} $Q$ is [now] empty:\\
\mbox{}~~~~do \textbf{B}\\
\mbox{}~~~~\textbf{IF} routine \textbf{B} ended by all operations being
done: exit [skipping \textbf{C} below]\\
do \textbf{C}
\end{quote}

This takes $O(f(n))$ time worst case, since each of the routines
do. We now want to argue that $|\A'''(S)| = O(|\A'(S)|)$ on any
sequence~$S$, by charging all work of~$\A'''$ to work of $\A'$ that has
been executed. There are five different types of actions of~$\A'''$
possible, with routine sequences as follows: \textbf{AB}, \textbf{ABC},
\textbf{AC}, \textbf{B}, and \textbf{BC}.

The action \textbf{B} starts and ends with $Q$ empty, and takes time
proportional to the work done on $\A'$, hence that work can be charged.
The action \textbf{BC} starts with $Q$ empty, ends with $Q$ non-empty,
takes time $O(f(n))$, and does $f(n)$ work on $\A'$, hence that work
can be charged.
The action \textbf{ABC} starts with $Q$ non-empty, ends with $Q$
non-empty, has $Q$ empty in the meantime, takes time $O(f(n))$ and
does $f(n)$ work on $\A'$, hence that work can be charged.
The action \textbf{AB} starts with $Q$ non-empty, ends with $Q$ empty,
takes time $O(f(n))$ but does possibly only $O(1)$ work on
$\A'$. However, the action (either \textbf{BC} or \textbf{ABC}) during
which $Q$ last turned from empty to non-empty did $f(n)$ work on $\A'$,
hence that work can be charged.

Remaining is the action \textbf{AC}, which has $Q$ non-empty from start
to end. Note that in the \textbf{A}~part, the $d\cdot f(n)$ BST
operations of $\A'''$ for $f(n) = \log n$ can be a couple of dequeues
(each taking $\log n$ operations), all of keys for which there are a
constant amount of $\A'$ work. Hence, we cannot charge the $\A'$ work on a
per-action basis, and a more elaborate charging argument is needed:
Consider a sequence of $t$ \textbf{AC} actions, following an action
(either \textbf{BC} or \textbf{ABC}) during which $Q$ last turned from
empty to non-empty. There have been exactly $t+1$ elements enqueued
(during \textbf{C} parts) since the queue was last empty, of which (at
least) the last is still present. Hence, during the $t$ \textbf{AC}
actions at most $t$ dequeues can have been done. Also, exactly $t$
restarts, $t$ key accesses via the balanced tree, and $t$ enqueues have
been done. All these sum up to at most $c t \cdot \log n$ BST operations
for $\A'''$, for some constant $c$. Let $c'$ be a constant such that
$f(n) \ge c' \cdot \log n$. At least $d t \cdot f(n) \ge d c'
t \cdot \log n$ BST operations have been done by $\A'''$ during the
$t$ actions, and those not included in the sum above must be $\A'$
process work. Thus, by choosing $d$ large enough that $d c' \ge
2c$, at least half of the work done must be $\A'$ process work. Hence, we
can charge all $\A'''$ work to $\A'$ work executed during the $t$ actions.

Summing up, over the entire sequence, all work of $\A'''$ can be charged
to executed work of $\A'$, with no work of $\A'$ being charged more than a
constant number of times.  Hence, $|\A'''(S)| = O(|\A'(S)|)$. By
$|\A'(S)| = O(|\A(S)|)$ from Theorem~\ref{thm:simulation}, the claim
$|\A'''(S)| = O(|\A(S)|)$ of Theorem~\ref{thm:online} follows.

Finally, as the queue has only room for $n$ elements, we need to
guarantee that the queue does not overflow. The queue overflows exactly
when $n$ \textbf{AC} actions in a row have taken place. By the argument
above, at least $1/2 \cdot d n \cdot f(n)$ $\A'$ work has been
executed, which for $d$ large enough leads to a contradiction with the
guarantee on $\A$'s performance.

\end{proof}

We note that one feature of BST algorithms not maintained by
Theorem~\ref{thm:online} is the exact amount of information stored in tree
nodes, besides the search key. For classical BST algorithms, this varies
from zero bits in Splay trees, over one bit in red-black trees, two bits in
AVL-trees, to $\Theta(\log n)$ bits in weight-balanced trees and treaps. Algorithm
$\A'''$ from Theorem~\ref{thm:online} always uses $\Theta(\log n)$ bits.


\bibliographystyle{abbrv}
\bibliography{splay}

\end{document}